\documentclass{article}
\pdfoutput=1

\usepackage{amsmath}
\usepackage{amssymb}
\usepackage{amsthm}
\usepackage[english]{babel}
\usepackage[backend=bibtex]{biblatex}
\usepackage[T1]{fontenc}
\usepackage[utf8]{inputenc}
\usepackage{lmodern}
\usepackage{csquotes}
\usepackage{thmtools}
\usepackage{hyperref}
\usepackage{complexity}
\usepackage{microtype}
\usepackage{textcomp}
\usepackage{tikz}

\LoadMicrotypeFile{cmr}
\SetProtrusion
    [load=lmr-T1]
    {encoding=T1, family=lmr}
    {
      \textquotedblright = {,1000},
      \textquotedblleft = {1000,},
      {'} = {,1000},
      {,} = {,1000},
      {:} = {,1000},
      {;} = {,1000},
      {.} = {,1000}
    }

\hypersetup{pdftitle={Polynomial-time kernel reductions}, pdfauthor={Jeffrey Finkelstein}}

\declaretheorem[numberwithin=section]{theorem}

\declaretheorem[numberlike=theorem]{lemma}
\declaretheorem[numberlike=theorem]{proposition}
\declaretheorem[numberlike=theorem]{corollary}
\declaretheorem[numberlike=theorem, style=definition, qed=\qedsymbol]{construction}
\declaretheorem[numberlike=theorem, style=definition]{definition}
\declaretheorem[numberlike=theorem, style=definition, name=Open problem]{openproblem}
\declaretheorem[numberlike=theorem, style=definition]{example}

\newenvironment{sketch}{\begin{proof}[Proof sketch]}{\end{proof}}

\newcommand{\email}[1]{\textlangle\href{mailto:#1}{\nolinkurl{#1}}\textrangle}
\newcommand{\plain}[1]{\text{ #1 }} 
\newcommand{\nkr}{\nleq^{P}_{ker}} 
\newcommand{\kr}{\leq^{P}_{ker}} 
\newcommand{\pot}{\leq^{P}_{pot}} 
\newcommand{\npot}{\nleq^{P}_{pot}} 
\newcommand{\krnt}{\leq_{ker}} 
\newcommand{\nkrnt}{\nleq_{ker}} 
\newcommand{\kri}{\leq^{P}_{ker,1}} 
\newcommand{\mor}{\leq^{P}_{m}} 
\newcommand{\mornt}{\leq_m}
\newcommand{\lb}{\left\{} 
\newcommand{\rb}{\right\}} 
\newcommand{\st}{\,\middle|\,} 
\newcommand{\defn}[1]{\emph{#1}} 
\newcommand{\pair}[2]{\langle#1,#2\rangle} 
\newcommand{\triple}[3]{\left\langle#1,#2,#3\right\rangle} 

\long\def\symbolfootnote#1{\begingroup%
\def\thefootnote{\fnsymbol{footnote}}\footnotetext{#1}\endgroup} 

\author{Jeffrey~Finkelstein \\ Boston University \\ \email{jeffreyf@bu.edu}
  \and Benjamin~Hescott \\ Tufts University \\ \email{hescott@cs.tufts.edu}}
\title{Polynomial-time kernel reductions}

\begin{document}

\maketitle

\begin{abstract}
  %
  Today, the computational complexity of equivalence problems such as the graph isomorphism problem and the Boolean formula equivalence problem remain only partially understood.
  One of the most important tools for determining the (relative) difficulty of a computational problem is the many-one reduction, which provides a way to encode an instance of one problem into an instance of another.
  In equivalence problems, the goal is to determine if a pair of strings is related, so a many-one reduction with access to the entire pair may be too powerful.
  A recently introduced type of reduction, the \emph{kernel reduction}, defined only on equivalence problems, allows the transformation of each string in the pair independently.
  Understanding the limitations of the kernel reduction as compared with the many-one reduction improves our understanding of the limitations of computers in solving problems of equivalence.
  We investigate not only these limitations, but also whether classes of equivalence problems have complete problems under kernel reductions.
  This paper provides a detailed collection of results about kernel reductions.

  %
  After exploring possible definitions of complexity classes of equivalence relations, we prove that polynomial time kernel reductions are strictly less powerful than polynomial time many-one reductions.
  We also provide sufficient conditions for complete problems under kernel reductions, show that completeness under kernel reductions can sometimes imply completeness under many-one reductions, and finally prove that equivalence problems of intermediate difficulty can exist under the right conditions.
  Though kernel reductions share some basic properties with many-one reductions, ultimately the number and size of equivalence classes can prevent the existence of a kernel reduction, regardless of the complexity of the equivalence problem.
  The most important open problem we leave unsolved is proving the unconditional existence of a complete problem under kernel reductions for some basic complexity classes that are well-known to have complete problems under many-one reductions.
\end{abstract}

\symbolfootnote{%
  Copyright 2010--2016 Jeffrey~Finkelstein and Benjamin~Hescott.

  This document is licensed under the Creative Commons Attribution-ShareAlike 4.0 International License, which is available at \mbox{\url{https://creativecommons.org/licenses/by-sa/4.0/}}.
  The \LaTeX{} markup that generated this document can be downloaded from its website at \mbox{\url{https://github.com/jfinkels/equivalence}}.
  The markup is distributed under the same license.
}

\section{Introduction}

The computational complexity of deciding whether two graphs are isomorphic has significant implications not only in computer science, but also in the computational forms of sciences such as chemistry, biology, and neuroscience.
One main technique for determining the complexity of the problem is showing how the difficulty of the problem relates to the difficulty of other known problems.
The relative difficulty of computational problems are often compared using the many-one reduction, a function by which we encode an instance of a problem as an instance of another problem.
In the case of the graph isomorphism problem, a many-one reduction from the graph isomorphism problem to, for example, the directed graph isomorphism problem allows the function computing the reduction to have access to both graphs in an instance of the problem.
However, access to both graphs is not necessary for computing the reduction; the function transforms each undirected graph independently into a directed graph.
In other words, the reduction is in reality defined on the domain of graphs, not on the domain of pairs of graphs.
This is a far more natural way to define reductions between problems of equivalence, and is furthermore a finer-grained comparison of the relative difficulty of the two computational problems.

The \emph{kernel reduction}, defined in \autocite[Definition~4.13]{fg11}, formally captures this notion of reduction among computational problems of equivalence involving independent transformation of each element of a pair.
This type of reduction has appeared previously under other names not only in this setting but also in more general settings (``Borel reduction''\kern-0.5em,\kern0.5em ``strong isomorphism reduction''\kern-0.5em,\kern0.5em ``strong equivalence reduction''\kern-0.5em,\kern0.5em ``relation reduction''\kern-0.5em,\kern0.5em ``component-wise reduction''\kern-0.5em,\kern0.5em etc.).
To the best of our knowledge, every known many-one reduction between problems of equivalence is really a kernel reduction (see, for an early example, the list of problems many-one reducible to graph isomorphism given in \autocite{bc79}).
However, kernel reductions seem less powerful than many-one reductions, since the former has access only to one element of a pair at a time.
What are the limitations of kernel reductions?

Some of our theorems adapt or clarify existing work in order to have simpler, self-contained, complexity-theoretic proofs of important theorems about kernel reductions.
In \autocite{fg11}, the authors ask whether kernel reductions and many-one reductions are provably different.
However, little beyond the definition is given there, other than the general idea that an imbalance in the number of equivalence classes of the two equivalence problems prevents the existence of a kernel reduction.
In computability theory, a similar type of reduction between equivalence problems has been well-studied by a series of recent papers (for example, \autocite{gg01, ff12, ffn12, chm12, imnn13, almnss14, mn14}).
However, these papers do not focus on efficiently computable reductions.
In \autocite{bcffm}, the authors provide a thorough treatment of not only the kernel reduction but also a generalization called the ``strong isomorphism reduction''\kern-0.5em.
Strong isomorphism reductions are themselves a special case of ``functorial reductions'' (a name borrowed from the language of category theory), which are reductions that make explicit the morphism between the category of objects being transformed.
These reductions were apparently defined in unpublished manuscripts \autocite{babai77} and \autocite{kucera76} (see \autocite[Section~15]{zkt85} for a contemporary definition, \autocite[Section~7]{babai14} for a more recent one).
Finally, the authors of \autocite{gz14} extended the work of \autocite{bcffm}, and in doing so, independently proved the main combinatorial idea used in this paper to examine the limitations of kernel reductions.
This paper, complementing that work, focuses mainly on completeness results.

We undertake a thorough investigation of the basic properties of kernel reductions, comparing them with the basic properties of many-one reductions.
The starting point for understanding many-one reductions is $\P$ and $\NP$, so we attempt to extend the definition from \autocite{fg11} of $\PEq$, the class of equivalence problems decidable in polynomial time, to the definition of the complexity class $\NPEq$ (\autoref{sec:definitions}).
We determine the limitations of kernel reductions; these appear to be combinatorial, not computational, in nature (\autoref{sec:limitations}).
We discover sufficient conditions for complete problems under kernel reductions in classes of equivalence problems (\autoref{sec:generalcompleteness}).
We compare the new notion of completeness under kernel reductions with the usual notion of completeness under many-one reductions (\autoref{sec:npeqcompleteness}).
Finally, as an analog to $\NP$-intermediary problems with respect to many-one reductions, we examine the possibility of $\NPEq$-intermediary problems with respect to kernel reductions (\autoref{sec:intermediary}).

\section{Preliminaries}
\label{sec:preliminaries}

The set of natural numbers (including $0$) is denoted $\mathbb{N}$, the set of integers is denoted $\mathbb{Z}$, and the set of positive integers is denoted $\mathbb{Z}^+$.

If $f\colon S\to T$ is a well-defined function and $S'\subseteq S$, then \defn{$f$ restricted to the domain $S'$} is the function $f'\colon S'\to T$ defined by $f'(x)=f(x)$ for all $x\in S'$.
We denote this restricted function on a smaller domain by $f|_{S'}$.
The \emph{image of $S'$}, denoted $f(S')$, is defined by $f(S') = \{f(s) \, | \, s \in S'\}$.

In this paper, $\Sigma$ denotes the binary alphabet $\{0, 1\}$.
$\Sigma^*$ is the set of all binary strings over the alphabet $\Sigma$ and $\Sigma^{\leq n}$ is the set $\lb w\in\Sigma^* \st |w|\leq n \rb$.
The empty string will be denoted by $\lambda$.
If $\sigma\in\Sigma$ then $\sigma^k$ is the string consisting of $k$ concatenated copies of the symbol $\sigma$.
If $x$ and $y$ are elements of $\Sigma^*$, then we denote by $\pair{x}{y}$ the \defn{pairwise encoding} of $x$ and $y$, which is itself an element of $\Sigma^*$.
In this paper, we will assume the reasonable pairwise encoding defined by $\pair{x}{y}=x_1x_1x_2x_2\cdots x_{|x|}x_{|x|}01y_1y_1y_2y_2\cdots y_{|y|}y_{|y|}$ for all $x$ and $y$ in $\Sigma^*$.
As usual, a \defn{language} over an alphabet $\Sigma$ is a subset of $\Sigma^*$.
The \defn{complement} of a language $L$ is $\Sigma^*\backslash L$, and is denoted $\overline{L}$.

The complexity classes $\P$, $\NP$, $\FP$ (polynomial-time computable functions), $\SKP$, $\PKP$, $\DKP$, and $\PSPACE$ have the usual definitions.
The set of words accepted by a Turing machine $M$ is denoted $L(M)$.
The \defn{complement} of a complexity class $\mathcal{C}$ is the set of complements of languages in $\mathcal{C}$, and is denoted $\coC$.

We say a Turing machine $M$ is a \defn{polynomially clocked Turing machine} if the description of $M$ includes a positive integer $k$ such that $M$ halts within time $kn^k$ on all inputs of length $n$.

If $L_1$ and $L_2$ are languages, we say that \defn{$L_1$ many-one reduces to $L_2$} if there exists a computable function $f$ such that $w \in L_1$ if and only if $f(w) \in L_2$.
We denote this by $L_1 \mornt L_2$.
If $f$ is computable in polynomial time, we denote this by $L_1 \mor L_2$.

A set $R \subseteq \Sigma^* \times \Sigma^*$ is an \defn{equivalence relation on $\Sigma^*$} if $R$ satisfies the following three properties.
\begin{itemize}
\item (reflexivity) For all $x \in \Sigma^*$, $(x,x)\in R$.
\item (symmetry) For all $x,y\in \Sigma^*$, $(x,y)\in R$ implies $(y,x)\in R$.
\item (transitivity) For all $x,y,z\in \Sigma^*$, $(x,y)\in R$ and $(y,z)\in R$ implies $(x,z)\in R$.
\end{itemize}
An equivalence relation $R$ can be encoded as a language by taking the pairwise encoding of each pair in $R$.
In this way we can study the computational complexity of classes of languages which represent equivalence relations.
In this paper we will abuse notation and write $\pair{x}{y}\in R$ for an equivalence relation $R$ on $\Sigma^*$, but what we really mean is $(x,y)\in R$ and $\pair{x}{y}\in L_R$, the language on the alphabet $\Sigma$ induced by $R$.

The \defn{equivalence class} of $x$ with respect to an equivalence relation $R$ on $\Sigma^*$ is $\lb y\in \Sigma^* \st (x,y)\in R \rb$.
It is denoted $[x]_R$, or if the context is clear, simply $[x]$.
Each element $x\in \Sigma^*$ is in exactly one equivalence class, so the equivalence classes of an equivalence relation on $\Sigma^*$ provide a partition of $\Sigma^*$.
Conversely, a partition of $\Sigma^*$ induces an equivalence relation on $\Sigma^*$ in which a pair of elements is in the relation if they are in the same block of the partition.

A \defn{complete invariant} for an equivalence relation $R$ on $\Sigma^*$ is a function $f\colon \Sigma^* \to \Sigma^*$ such that for each $x$ and $y$ in $\Sigma^*$, we have $(x, y) \in R$ if and only if $f(x) = f(y)$.
(A \emph{canonical form} for an equivalence relation is a complete invariant satisfying the additional requirement that $f(x) \in [x]_R$; canonical forms, though important, do not appear in this paper.)
In \autoref{sec:definitions} we will define generalizations of the complete invariant which accept as input an additional witness to the equivalence of $x$ and $y$.

\defn{$\PEq$} is the class of equivalence relations for which membership can be decided by a Turing machine running in deterministic polynomial time.
\defn{$\NPEq$} is the class of equivalence relations for which membership can be decided by a Turing machine running in non-deterministic polynomial time.
In other words, $\PEq$ is the set of (languages induced by) equivalence relations which are in \P, and $\NPEq$ is the set of (languages induced by) equivalence relations which are in \NP.
In general, the class \defn{$\CEq$} is the class of languages induced by equivalence relations which are in the complexity class $\mathcal{C}$.
As usual, $\PEq\subseteq\NPEq$.

We now require a natural notion of reduction among equivalence relations.
If $R$ and $S$ are equivalence relations on $\Sigma^*$, we say $R$ \defn{kernel reduces to} $S$ if there exists a computable $f\colon\Sigma^*\to\Sigma^*$ such that $\forall x,y\in\Sigma^*$, $\pair{x}{y}\in R\iff \pair{f(x)}{f(y)}\in S$.
We denote this by $R\krnt S$.
If $f$ is computable in polynomial time, then we say $R$ \defn{polynomial-time kernel reduces to} $S$ and use the notation $R\kr S$.

Notice the difference between a kernel reduction and a many-one reduction: a kernel reduction maps $\pair{x}{y}\in R$ to $\pair{f(x)}{f(y)}\in S$, whereas a many-one reduction maps $\pair{x}{y}\in R$ to $f(\pair{x}{y})\in S$, for some polynomial-time computable function $f$.
Informally, a function which computes a many-one reduction has access to both $x$ and $y$ but a function which computes a kernel reduction has access to only one of $x$ and $y$ at a time.
Since it is more restrictive, a kernel reduction induces a many-one reduction (namely the function $\pair{x}{y} \mapsto \pair{f(x)}{f(y)}$).
Still, kernel reductions compose just as many-one reductions do, and $\NPEq$ is closed under polynomial-time kernel reductions, allowing us to adapt existing complexity theoretic analysis to the study of complexity of equivalence relations.

As an analog to polynomial-time many-one completeness in \NP, we define a similar notion of completeness under polynomial-time kernel reductions in \NPEq.
An equivalence relation $S$ is \defn{\NPEq-hard} if for all $R\in\NPEq$, $R\kr S$.
If $S$ is also in \NPEq, then it is \defn{\NPEq-complete}.
If $S$ is \NPEq-complete, we sometimes say that $S$ is \defn{complete under $\kr$ reductions in \NPEq}.
Generally, an equivalence relation $S$ is \defn{$\CEq$-hard} if for all $R\in\CEq$, $R\kr S$, and \defn{$\CEq$-complete} if it is additionally in $\CEq$.

\section{Definitions of \texorpdfstring{\NPEq}{NPEq}}
\label{sec:definitions}
%
The main property of languages in $\NP$ is that membership in each language is verifiable in polynomial time, given a witness to the membership.
Many important equivalence problems are in $\NP$, and some are even $\NP$-complete, but these are complete under traditional many-one reductions, not kernel reductions.
We wish to define $\NPEq$ as the class of equivalence problems that are efficiently verifiable, just as we define $\NP$ as the class of all computational problems.
One way to define $\NPEq$ is simply as the subclass of $\NP$ that includes only equivalence problems.
This section provides some other possible definitions based on our intuition about ``efficiently verifiable'' equivalence problems and compares those definitions.

%
We show that the alternative definitions of $\NPEq$ form a hierarchy below $\NPEq$ as defined above.
In other words, $\NPEq$ is the most general class of efficiently verifiable equivalence problems.
When attempting to prove that there are complete problems in $\NPEq$ under kernel reductions, we must therefore use this most general definition.
It remains to show whether any of the (non-equal) alternative definitions are distinct, and whether any of them has a complete problem under kernel reductions.

The first definition is the analog of the fundamental definition of $\NP$; it is the formal definition of the class $\NPEq$ introduced in the previous section.

\begin{definition}\label{def:npeq1}
  An equivalence relation $R$ is in $\NPEq$ if there is a polynomial $p$ and a nondeterministic Turing machine $N$ such that for each $x$ and $y$, the machine $N$ halts in time $p(\left|\pair{x}{y}\right|)$ and
  \begin{displaymath}
    \pair{x}{y} \in R \iff N(\pair{x}{y}) \text{ accepts}.
  \end{displaymath}
\end{definition}

Just as there is a definition of $\NP$ using polynomial-time verifiers, there is an equivalent definition for $\NPEq$ using polynomial-time verifiers.
However, this definition feels a bit unnatural when dealing with equivalence relations, since the witness language would be a relation (of the form ``$(x, y)$ relates to $w$''), but not an \emph{equivalence} relation.
The next two definitions attempt to require that the witness language is itself an equivalence relation, instead of an arbitrary language in $\P$.
Each of these ``witness equivalence relations'' is a set of pairs of pairs, in which each inner pair includes a witness string.

\begin{definition}\label{def:npeq3}
  Suppose $R'$ is an equivalence relation in $\PEq$.
  An equivalence relation $R$ is a \emph{two-witness projection} of $R'$ if for each binary string $x$ and $y$,
  \begin{displaymath}
    \pair{x}{y} \in R \iff \exists w_x, w_y \colon \pair{\pair{x}{w_x}}{\pair{y}{w_y}} \in R'\kern-0.3em,
  \end{displaymath}
  where $|w_x|$ is polynomially bounded in $|x|$ and $|w_y|$ is polynomially bounded in $|y|$.
  The class $\Proj_2$ is the collection of all two-witness projections of equivalence relations in $\PEq$.
\end{definition}

\begin{definition}\label{def:npeq4}
  Suppose $R'$ is an equivalence relation in $\PEq$.
  An equivalence relation $R$ is a \emph{one-witness projection} of $R'$ if for each binary string $x$ and $y$,
  \begin{displaymath}
    \pair{x}{y} \in R \iff \exists w \colon \pair{\pair{x}{w}}{\pair{y}{w}} \in R'\kern-0.3em,
  \end{displaymath}
  where $|w|$ is polynomially bounded in $\min(|x|, |y|)$.
  The class $\Proj_1$ is the collection of all one-witness projections of equivalence relations in $\PEq$.
\end{definition}

The next two definitions attempt to allow the possibility of not just a simple string which witnesses the equivalence of $x$ and $y$, but a ``witness function'' which may map $x$ and $y$, along with witness strings, to an equivalence relation in \PEq.

\begin{definition}\label{def:npeq5}
  Suppose $R$ and $R'$ are equivalence relations.
  A function $f$ is a \emph{nondeterministic polynomial-time two-witness kernel reduction} from $R$ to $R'$ if $f$ is in $\FP$ and for each binary string $x$ and $y$,
  \begin{displaymath}
    \pair{x}{y} \in R \iff \exists w_x, w_y \colon \pair{f(x, w_x)}{f(y, w_y)} \in R'\kern-0.3em,
  \end{displaymath}
  where $|w_x|$ is polynomially bounded in $|x|$ and $|w_y|$ is polynomially bounded in $|y|$.
  The class $\Cl_2$ is the closure of $\PEq$ under these reductions.
\end{definition}

\begin{definition}\label{def:npeq6}
  Suppose $R$ and $R'$ are equivalence relations.
  A function $f$ is a \emph{nondeterministic polynomial-time one-witness kernel reduction} from $R$ to $R'$ if $f$ is in $\FP$ and for each binary string $x$ and $y$,
  \begin{displaymath}
    \pair{x}{y} \in R \iff \exists w \colon \pair{f(x, w)}{f(y, w)} \in R'\kern-0.3em,
  \end{displaymath}
  where $|w|$ is polynomially bounded in $\min(|x|, |y|)$.
  The class $\Cl_1$ is the closure of $\PEq$ under these reductions.
\end{definition}

The final two definitions attempt to describe equivalence relations for which there is a ``witnessed complete invariant''\kern-0.5em,\kern0.5em which maps equivalent strings to \emph{equal} strings when given access to some witness of their equivalence.

\begin{definition}\label{def:npeq7}
  Suppose $R$ is an equivalence relation.
  A function $f$ is a \emph{nondeterministic polynomial-time two-witness complete invariant} for $R$ if $f$ is in $\FP$ and for each binary string $x$ and $y$,
  \begin{displaymath}
    \pair{x}{y} \in R \iff \exists w_x, w_y \colon f(x, w_x) = f(y, w_y),
  \end{displaymath}
  where $|w_x|$ is polynomially bounded in $|x|$ and $|w_y|$ is polynomially bounded in $|y|$.
  The class $\NKer_2$ is the collection of all equivalence relations that admit such a function.
\end{definition}

\begin{definition}\label{def:npeq8}
  Suppose $R$ is an equivalence relation.
  A function $f$ is a \emph{nondeterministic polynomial-time one-witness complete invariant} for $R$ if $f$ is in $\FP$ and for each binary string $x$ and $y$,
  \begin{displaymath}
    \pair{x}{y} \in R \iff \exists w \colon f(x, w) = f(y, w),
  \end{displaymath}
  where $|w|$ is polynomially bounded in $\min(|x|, |y|)$.
  The class $\NKer_1$ is the collection of all equivalence relations that admit such a function.
\end{definition}

The definitions of these complexity classes yield a chain of inclusions beginning with $\NKer_1$ and terminating with $\NPEq$.
\begin{theorem}\label{thm:definitions}
  $\NKer_1 = \Cl_1 = \Proj_1 \subseteq \NKer_2 \subseteq \Cl_2 = \Proj_2 \subseteq \NPEq$.
\end{theorem}
\begin{sketch}
  $\Proj_1 \subseteq \Cl_1$ by choosing the kernel reduction $f$ to be the identity function.
  $\Cl_1 \subseteq \NKer_1$ by choosing the complete invariant $f'$ to be
  \begin{equation*}
    f'(x, w') =
    \begin{cases}
      w'0 & \text{if } \pair{f(x, v)}{f(y, v)} \in R', \text{ where } w' = (y, v) \\
      x1 & \text{otherwise},
    \end{cases}
  \end{equation*}
  where $f$ is the kernel reduction.
  $\NKer_1 \subseteq \Proj_1$ by choosing $R'$ to be the equality relation after an application of the complete invariant $f$ to both the left pair and the right pair in the relation.

  $\NKer_1 \subseteq \NKer_2$ by choosing both $w_x$ and $w_y$ to be the witness $w$.
  $\NKer_2 \subseteq \Cl_2$ by choosing $R'$ to be the equality relation.

  $\Proj_2 \subseteq \Cl_2$ by choosing $f$ to be the identity function.
  $\Cl_2 \subseteq \Proj_2$ by hardcoding the function $f$ into the relation $R'$.

  $\Proj_2 \subseteq \NPEq$ by defining $N$ to nondeterministically choose $w_x$ and $w_y$ then verify that $\pair{x}{w_x}$ and $\pair{y}{w_y}$ are related under $R'$.
\end{sketch}

We are unable to show $\Cl_2 \subseteq \NKer_2$ using the technique that shows $\Cl_1 \subseteq \NKer_1$ because the complete invariant $f'$ cannot access both of the necessary witnesses for the kernel reduction $f$ in a symmetric way.
The best we can do is show this inclusion under an assumption.

Our one-witness and two-witness complete invariants are generalizations of the deterministic complete invariant, as defined in \autoref{sec:preliminaries}.
In \autocite{fg11}, the authors define the class $\Ker$ as the set of all equivalence relations $R$ that have a polynomial-time computable complete invariant.
They provide evidence that $\Ker$ and $\PEq$ are different by showing that equality of the two classes implies some unlikely collapses in ``higher'' complexity classes.
Unfortunately, we are only able to show that $\Cl_2 \subseteq \NKer_2$ under the assumption that $\Ker = \PEq$.

\begin{corollary}
  If $\Ker = \PEq$, then $\NKer_2 = \Cl_2 = \Proj_2$.
\end{corollary}
\begin{proof}
  By the previous theorem, it suffices to show $\Proj_2 \subseteq \NKer_2$.
  Suppose $R \in \Proj_2$, so there is an $R' \in \PEq$ such that $\pair{x}{y} \in R$ if and only if there are $w_x$ and $w_y$ such that $\pair{\pair{x}{w_x}}{\pair{y}{w_y}} \in R'$.
  Since $\Ker = \PEq$, there is a function $f \in \FP$ such that $\pair{\pair{x}{w_x}}{\pair{y}{w_y}} \in R'$ if and only if $f(x, w_x) = f(y, w_y)$.
  Thus there is a function $f$ such that $\pair{x}{y} \in R$ if and only if there are $w_x$ and $w_y$ such that $f(x, w_x) = f(y, w_y)$.
  Therefore $R \in \NKer_2$.
\end{proof}


\section{Limitations of kernel reductions}\label{sec:limitations}
%
Can a kernel reduction be used anywhere a many-one reduction can be used?
If so, a function with access to one element of a pair would be exactly as powerful as a function with access to both elements of a pair; our intuition is that this is unlikely.
This section proves that polynomial-time kernel reductions are strictly weaker than polynomial-time many-one reductions.

%
We find that a bound on the size of the image of a kernel reduction implies that the function can only access a finite number of equivalence classes.
Constructing equivalence relations so that there is an imbalance in the number of equivalence classes with respect to any fixed function suffices to show that no polynomial-time kernel reduction can exist between the two.
Thus, we conclude that polynomial-time kernel reductions are more restrictive than polynomial-time many-one reductions.
This will be important for \autoref{sec:generalcompleteness} as well, since it means that completeness under kernel reductions is distinct from completeness under many-one reductions.

We adopt and extend the notation $\#R$ from \autocite{bcffm} to denote the number of equivalence classes in an equivalence relation $R$.

\begin{definition}[{\autocite[Section~5]{bcffm}}]
  Suppose $R$ is an equivalence relation on $\Sigma^*$.
  Let $\#R(n) = \left|\left\{[x]_R \, \middle| \, x \in \Sigma^{\leq n}\right\}\right|$, or in other words, $\#R(n)$ is the number of equivalence classes in $R$ for strings of length at most $n$.
  Let $\#R = \max\limits_{n \in \mathbb{N}} \#R(n)$ if the maximum exists, or in other words, $\#R$ is the number of equivalence classes in $R$.
\end{definition}

As first stated in \autocite{fg11}, if the number of equivalence classes in $R$ is greater than the number of equivalence classes in $S$, then no kernel reduction can exist (regardless of any time or space bounds on the function computing the reduction).
For completeness, we prove this basic fact in \autoref{prop:noreduction} below.
However, a many-one reduction can overcome this restriction by having access to both strings in the pair.
Before proving that, we require the following lemma showing that kernel reductions must preserve ``related-ness'' of pairs of elements by mapping equivalence classes in $R$ to equivalence classes in $S$.
(The proof, omitted here, is a straightforward application of the definitions.)

\begin{lemma}\label{lem:image}
  Suppose $R$ and $S$ are equivalence relations on $\Sigma^*$.
  Suppose $R \krnt S$ and $f$ is the function computing the kernel reduction.
  Let $\hat{f}$ denote the function defined by $\hat{f}([x]_R) = [f(x)]_S$, for all equivalence classes $[x]_R$ in $R$.
  Then
  \begin{itemize}
  \item $\hat{f}$ is injective,
  \item $f([w]_R) \subseteq \hat{f}([w]_R)$ for any $w \in \Sigma^*$.
  \end{itemize}
\end{lemma}


\begin{proposition}\label{prop:noreduction}
  Let $R$ and $S$ be equivalence relations on $\Sigma^*$.
  If $\#R > \#S$, then $R \nkrnt S$.

  Furthermore, suppose $\#R = n$ and $\#S = m$, and suppose $m \geq 2$.
  Let $r_1, \dotsc, r_n$ and $s_1, \dotsc, s_m$ denote representatives of the equivalence classes in $R$ and $S$, respectively.
  If the problem of deciding whether $x \in [r_i]_R$ for any $x \in \Sigma^*$ is recognizable, then $R \mornt S$.
\end{proposition}
\begin{proof}
  Assume that $R \krnt S$.
  By \autoref{lem:image}, the function mapping equivalence classes in $R$ to equivalence classes in $S$ induced by the kernel reduction is injective.
  However, this violates the pigeonhole principle.
  Therefore no kernel reduction exists from $R$ to $S$.

  On the other hand, there is a many-one reduction from $R$ to $S$.
  First, suppose $S$ has $m$ equivalence classes and let $s_1, \dotsc, s_m$ be representatives of each equivalence class in $S$.
  On input $\pair{x}{y}$, for each $i \in \{1, \dotsc, n\}$ in parallel, determine if $x \in [r_i]_R$ and $y \in [r_i]_R$ (also in parallel).
  If $x$ and $y$ are both in $[r_i]_R$ for some $i$, output $\pair{s_i}{s_i}$, otherwise output $\pair{s_1}{s_2}$.

  Since each string must be in exactly one of the equivalence classes of $R$, this function must halt when searching for the equivalence class for the strings $x$ and $y$.
  If $\pair{x}{y} \in R$, then they are in the same equivalence class of $R$ and hence the function will output $\pair{s_i}{s_i}$, which is in $S$ by the reflexivity of $S$.
  If $\pair{x}{y} \notin R$, then they are in different equivalence classes and hence the function will output $\pair{s_1}{s_2}$, which is not in $S$ because $[s_1]_S \neq [s_2]_S$ by hypothesis.
  Therefore this function is a computable many-one reduction from $R$ to $S$.
\end{proof}

\begin{example}
  Let $R = \mathbb{Z} / 3 \mathbb{Z}$ and $S = \mathbb{Z} / 2 \mathbb{Z}$.
  Then $R \mornt S$ but $R \nkrnt S$.
\end{example}

As seen in \autoref{prop:noreduction}, for equivalence relations $R$ and $S$ with a finite number of equivalence classes, a kernel reduction from $R$ to $S$ can only exist if the number of equivalence classes in $R$ is at most the number of equivalence classes in $S$.
However, most ``interesting'' equivalence relations have an infinite number of equivalence classes.
In \autocite[Section~4]{fg11}, the authors ask if there are such equivalence relations ``of the same densities [that is, density of equivalence classes] on which kernel reduction and [many-one] reduction differ''\kern-0.5em.
\autocite[Theorem~5.1]{bcffm} (see also \autocite[Remark~5.2]{bcffm}) answers this question affirmatively, providing an infinite antichain of equivalence relations that are equivalent under polynomial-time many-one reductions but otherwise incomparable under polynomial-time ``strong isomorphism reductions'' (proven in \autocite[Section~7]{bcffm} to be equivalent to polynomial-time kernel reductions).
We will provide a simple proof of a special case of \autocite[Theorem~5.1]{bcffm}, showing that an imbalance in the density of equivalence classes prevents a kernel reduction.
This proof is valuable because it requires only knowledge of basic computational complexity theory and not knowledge of Boolean algebras, descriptive set theory, or other mathematical logic.

First we show that an equivalence relation dense in equivalence classes cannot be reduced to one sparse in equivalence classes.
We emphasize that our result does not concern the sparseness of strings in a language, but the sparseness of equivalence classes in an equivalence relation.
This complements the work on ``potential reducibility'' defined in \autocite[Section~5]{bcffm}.

\begin{definition}[{\autocite[Definition~7.2]{bcffm}}]
  Let $R$ and $S$ be equivalence relations on $\Sigma^*$.
  $R$ is \defn{potentially reducible} to $S$, denoted $R\pot S$, if there exists a polynomial $p$ such that for all $n\in\mathbb{N}$, $\#R(n)\leq \#S(p(n))$.
\end{definition}

It follows from the definitions that for any equivalence relations $R$ and $S$, $R\kr S\implies R\pot S$, and hence $R\npot S \implies R\nkr S$ (this is stated and proven explicitly in \autocite[Lemma~5.5]{bcffm}).
As an analog to traditional sparse languages, we provide a definition of ``kernel sparsity''\kern-0.5em, and show its application to determining potential reducibility and hence kernel reducibility.

\begin{definition}
  An equivalence relation $R$ on $\Sigma^*$ is \defn{kernel sparse} if there exists a polynomial $p$ such that for all $n\in\mathbb{N}$, $\#R(n)\leq p(n)$.
  In other words, the number of equivalence classes in $R$ for strings of length at most $n$ is bounded above by a polynomial in $n$.

  An equivalence relation is \defn{kernel dense} if it is not kernel sparse.
  Formally, if for all polynomials $p$ there exists an $n\in\mathbb{N}$ such that $\#R(n)>p(n)$.
  In other words, the number of equivalence classes in $R$ for strings of length at most $n$ is greater than any polynomial in $n$.
\end{definition}

These definitions allow us to provide the following very natural proposition.
Intuitively, it states that an equivalence relation with many closely packed equivalence classes cannot reduce (under polynomially bounded notions of reduction) to an equivalence relation with few but widely spaced equivalence classes.
This idea is stated without proof in \autocite[Section~4]{fg11}, so we provide it here for completeness.
It is also essentially a special case of \autocite[Lemma~2.3]{gz14}, developed independently of that paper.

\begin{theorem}\label{thm:density}
  Let $R$ and $S$ be equivalence relations on $\Sigma^*$.
  If $R$ is kernel dense and $S$ is kernel sparse, then $R\nkr S$.
\end{theorem}
\begin{proof}
  That $R\npot S$ implies $R\nkr S$ was already stated in the text preceding this theorem, so it suffices to show that $R\npot S$.

  Assume that $R\pot S$ with the intention of producing a contradiction.
  Let $p$ be a polynomial such that $\#R(n) \leq \#S(p(n))$ (this is the definition of potential reducibility).
  Let $q$ be a polynomial such that $\#S(n) < q(n)$ for each natural number $n$ (this is the definition of kernel sparse).
  Substituting $p(n)$ for $n$ in this inequality yields the inequality $\#S(p(n)) < q(p(n))$, which is a polynomial in $n$.
  Let $r = q \circ p$.

  Let $n_0$ be a natural number such that $\#R(n_0) > r(n_0)$, by the definition of kernel sparsity.
  Since $\#S(p(n_0)) \leq r(n_0)$, we have $\#R(n_0) > \#S(p(n_0))$.
  In other words, there are more equivalence classes in $R$ for strings up to length $n_0$ than there are in $S$ for strings up to length $p(n_0)$.
  By the pigeonhole principle, we conclude that $R$ cannot potentially reduce to $S$, because the number of equivalence classes in $R$ for strings up to length $n_0$ is too great compared to the number of equivalence classes in $S$ for strings up to length $p(n_0)$.
  This is a contradiction with the assumption that $R\pot S$.
  We have shown this for arbitrary polynomials (which came from the definitions of potential reducibility and kernel sparsity), so we can conclude that the result holds for all equivalence relations $R$ and $S$ that are kernel dense and kernel sparse, respectively.
\end{proof}

\begin{example}
  Consider the equality relation and the ``equal lengths'' relation (that is, $x$ relates to $y$ if $|x| = |y|$).
  The equality relation is kernel dense, since there are $2^n$ equivalence classes for strings of length at most $n$ (one for each string).
  The ``equal lengths'' relation is kernel sparse, since there are $n + 1$ equivalence classes for strings of length at most $n$ (one for each length, including length $0$).
  Therefore there is no polynomial-time kernel reduction from the equality relation to the ``equal lengths'' relation.
\end{example}

This places a strong restriction on equivalence relations that are hard (or complete) under polynomial-time kernel reductions: they cannot be kernel sparse.
This means that the equality relation, the densest possible equivalence relation with an exponential number of equivalence classes at each length, is a troublemaker in every complexity class that contains it.

\begin{corollary}
  Let $\CEq$ be a complexity class of equivalence relations containing the equality relation $R_{eq}$.
  If an equivalence relation $R$ is kernel sparse, then it is not $\CEq$-hard.
\end{corollary}

Polynomial-time many-one reductions are more powerful than polynomial-time kernel reductions because the former are not subject to restrictions on numbers of equivalence classes as in \autoref{prop:noreduction} and \autoref{thm:density}.
The idea behind \autoref{thm:density} leads to a construction of equivalence relations $R$ and $S$ between which there is a polynomial-time many-one reduction but no polynomial-time kernel reduction.

\begin{construction}\label{con:rands}
  Let $f_1, f_2, \dotsc$ be an enumeration of all polynomial-time computable functions.
  Assume, without loss of generality, that for all positive integers $i$, function $f_i$ runs in time $p_i(n)$, where $p_i(n) = i n^i$ for all positive integers $n$.

  Suppose $n$ is a positive integer.
  Define $R_n$ as the set of all strings of length $n$, except $R_1$, which also includes the string of length $0$.
  Define $S_n$ as the set of all strings $s$ satisfying the inequality $p_n(n) + 1 \leq |s| \leq p_{n + 1}(n + 1)$, except $S_1$, which includes all strings of length at most $p_2(2)$.

  Define sets $R$ and $S$ as
  \begin{equation*}
    R = \bigcup_{n \in \mathbb{Z}^+} R_n \times R_n \text{ and } S = \bigcup_{n \in \mathbb{Z}^+} S_n \times S_n.\qedhere
  \end{equation*}
\end{construction}

\begin{lemma}
  $R$ and $S$ are equivalence relations.
\end{lemma}
\begin{proof}
  $R$ and $S$ are equivalence relations if $\{R_n\}_{n \in \mathbb{Z}^+}$ and $\{S_n\}_{n \in \mathbb{Z}^+}$ are valid partitions of $\Sigma^*$, so it suffices to show that the union of each collection includes all nonempty strings in $\Sigma^*$ and that each collection is pairwise disjoint.

  For $\{R_n\}_n$, any string of length $n$ is in $R_n$, so $\Sigma^* \subseteq \cup_n R_n$.
  If $m$ and $n$ are distinct positive integers, no string can have both length $m$ and length $n$, so $R_m \cap R_n = \emptyset$.
  Hence $\{R_n\}_n$ is a valid partition.

  For $\{S_n\}_n$, for any string $x$, there is an $n$ such that $p_n(n) + 1 \leq |x| \leq p_{n + 1}(n + 1)$, so every string in $\Sigma^*$ is in some $S_n$.
  To show pairwise disjointness, suppose $m$ and $n$ are distinct positive integers and assume without loss of generality that $m < n$, or in other words, that $m + 1 \leq n$.
  Then $p_{m + 1}(m + 1) \leq p_n(n) < p_n(n) + 1$, so no string of length at most $p_{m + 1}(m + 1)$ can also have length at least $p_n(n) + 1$.
  Hence, $S_m$ and $S_n$ are disjoint.
  Thus, $\{S_n\}_n$ is a valid partition.

  Since both collections are valid partitions, the relations $R$ and $S$ are both equivalence relations.
\end{proof}

Again, this is a special case of \autocite[Theorem~5.1]{bcffm}, but has a much simpler proof and sufficiently demonstrates that polynomial-time kernel reductions and polynomial-time many-one reductions are different.

\begin{theorem}\label{thm:different}
  There are equivalence relations $R$ and $S$ such that $R \mor S$ but $R \nkr S$.
  Furthermore, $R$ and $S$ are in $\NCOneEq$.
\end{theorem}

The main idea behind this theorem is that no matter which polynomial-time function we consider as a possible kernel reduction, the number of equivalence classes in $R$ is greater than the number of equivalence classes in $S$, for sufficiently large strings.
\autoref{thm:density} doesn't apply in this setting because both $R$ and $S$ are kernel sparse.
Since we have carefully constructed these sets, $S$ is more kernel sparse than $R$.
This basic idea was presented independently in \autocite[Lemma~2.3]{gz14}.

Though it is not explicitly stated here, this theorem can be generalized to kernel reductions with other (non-polynomial) time bounds in a straightforward manner.

\begin{proof}[Proof of \autoref{thm:different}]
  Let $R$ and $S$ be the equivalence relations in \autoref{con:rands}.
  The following function is a polynomial-time many-one reduction from $R$ to $S$.
  On input $\pair{x}{y}$, if $|x| = |y|$ (or if $|x|$ and $|y|$ are both in $\{0, 1\}$), output $\pair{a}{a}$, otherwise output $\pair{a}{b}$, where $a$ is a string in $S_1$ and $b$ is a string in $S_2$.
  Computing and comparing the lengths of $x$ and $y$ can be done in linear time and writing the output requires only a constant number of steps, since the lengths of $a$ and $b$ are independent of the lengths of $x$ and $y$.
  The correctness of the reduction follows from the fact that $a$ and $b$ are in different equivalence classes.
  Therefore there is a polynomial-time many-one reduction from $R$ to $S$.

  Now assume with the intention of producing a contradiction that there is a polynomial-time kernel reduction from $R$ to $S$.
  Since $f_1, f_2, \dotsc$ is an enumeration of all polynomial-time computable functions, the reduction from $R$ to $S$ is $f_n$, with running time $p_n$, for some positive integer $n$.
  Consider a string $x$ of length $n + 1$ (for example, $x = 1^{n + 1}$); $x$ is in equivalence class $R_{n + 1}$.
  Since the running time of $f_n$ is $p_n$, the length of $f_n(x)$ is at most $p_n(n + 1)$.
  Since $p_1, p_2, \dotsc$ is an increasing sequence (in the sense that $p_j(n) < p_{j + 1}(n)$ for all natural numbers $n$ and all positive integers $j$), we have $p_n(n + 1) < p_{n + 1}(n + 1) < p_{n + 1}(n + 1) + 1$.
  By the construction of $R$ and $S$, we have $\#R(n + 1) = n + 1$ and $\#S(p_n(n + 1)) \leq \#S(p_{n + 1}(n + 1)) = n$ (for an illustration, see \autoref{fig:sparse}).
  By the pigeonhole principle, there must be two strings $x$ and $y$ of length at most $n + 1$ in different equivalence classes of $R$ whose image under $f_n$ is in the same equivalence class of $S$.
  Since $\pair{x}{y} \notin R$ if and only if $\pair{f(x)}{f(y)} \notin S$, this is a contradiction.
  Therefore $R \nkr S$.

  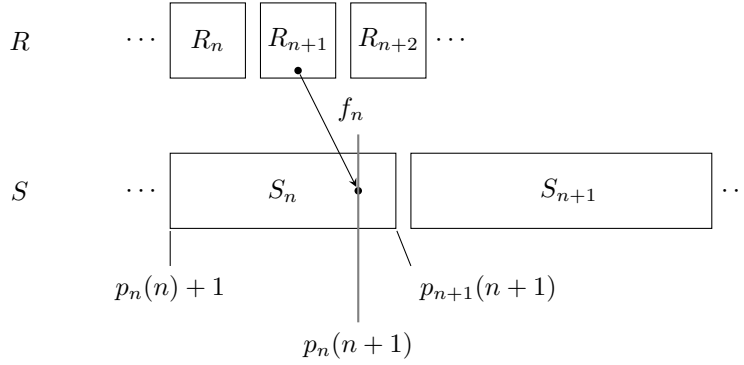
\begin{figure}
    \caption{\label{fig:sparse}For a fixed kernel reduction $f_n$ running in time $p_n$, the image of a string of length $n + 1$ can only be a string of length at most $p_n(n + 1)$.
      The number of equivalence classes in $R$ for strings of length $n + 1$ is greater than the number of equivalence classes in $S$ for strings of length $p_n(n + 1)$ for each polynomial $p_n$.}
    \begin{center}
      \begin{tikzpicture}[yscale=0.5]
        \begin{scope}[shift={(0, 5)}]
          \node at (0, 1) {$R$};

          \draw (2, 0)
          ++(0, 1) node[anchor=east] {$\dotsb$} ++(0, -1)
          rectangle ++(1, 2)
          ++(0.2, -2)
          rectangle ++(1, 2)
          ++(0.2, -2)
          rectangle ++(1, 2)
          ++(0, -1)
          node[anchor=west] {$\dotsb$}
          ;

          \path
          (2.5, 1) node {$R_n$}
          ++(1.2, 0) node {$R_{n + 1}$}
          ++(1.2, 0) node {$R_{n + 2}$}
          ;
        \end{scope}

        \begin{scope}[shift={(0, 1)}]
          \draw
          (4.5, 1) node[shape=circle, fill=black, inner sep=1pt] (source) {}
          ++(-0.8, 3.2) node[shape=circle, fill=black, inner sep=1pt] (target) {};
          \path[<-, >=stealth] (source) edge node[anchor=south west] {$f_n$} (target);
        \end{scope}

        \begin{scope}[shift={(0, 1)}]
          \node at (0, 1) {$S$};

          \draw (2, 0)
          ++(0, 1) node[anchor=east] {$\dotsb$} ++(0, -1)
          rectangle ++(3, 2)
          ++(0.2, -2)
          rectangle ++(4, 2)
          ++(0, -1)
          node[anchor=west] {$\dotsb$}
          ;

          \path
          (3.5, 1) node {$S_n$}
          ++(3.8, 0) node {$S_{n + 1}$}
          ;
        \end{scope}

        \begin{scope}
          \draw[shorten >=1pt]
          (2, 0) node[anchor=north] (pn) {$p_n(n) + 1$}
          -- ++(0, 1);
          \draw[thick, color=gray] (4.5, 3.5)
          -- ++(0, -5)
          node[anchor=north, color=black] (pnplus1) {$p_n(n + 1)$}
          ;
          \draw[shorten >=1pt]
          (5.2, 0) node[anchor=north west, color=black] (pnplus1nplus1) {$p_{n + 1}(n + 1)$}
          -- ++(-0.2, 1);
          ;
        \end{scope}
      \end{tikzpicture}
    \end{center}
  \end{figure}

  Finally, we show that $R$ and $S$ are in $\NC^1$.
  Deciding whether two strings have the same length is trivial, so $R$ is certainly in $\NC^1$.
  To decide $S$, we compute the index $i$ of the equivalence class $S_i$ containing the string $x$ and the index $j$ of the equivalence class $S_j$ containing the string $y$, then compare them for equality.
  Computing the index $i$ of the equivalence class of a string $x$ of length $n$ can be performed as follows.
  First, compute in parallel the values $p_1(1), \dotsc, p_{n + 1}(n + 1)$.
  Since each $p_i$ is increasing, $n$ is definitely smaller than $p_{n + 1}(n + 1)$.
  Computing the exponentiation of $O(n)$ pairs of strings of length $O(n)$ each can be performed by a $\TC^0$ circuit, and $\TC^0 \subseteq \NC^1$.
  Next, for each $i \in \{1, \dotsc, n\}$ in parallel, decide if $p_i(i) + 1 \leq n \leq p_{i + 1}(i + 1)$, thereby determining whether the input is in $S_i$.
  These comparisons can be performed by an $\NC^1$ circuit.
  Finally, use $O(\log n)$ single-bit multiplexers in parallel to output the index (in binary) of the sole equivalence class $S_i$ containing $x$.
  A single-bit multiplexer for $O(\log n)$ input bits can be implemented by a circuit of size $O(\log n)$ and depth $O(\log \log n)$ \autocite[Lemma~2.5.5]{savage98}, so this phase of the computation can be performed by an $\NC^1$ circuit.
  Computing the indices $i$ and $j$ for the two input strings can be performed in parallel, and the final comparison for equality of $i$ and $j$ adds only $O(\log n)$ depth to the circuit.
  Therefore $S \in \NC^1$.
\end{proof}

\section
    [Conditions for complete problems under polynomial-time kernel reductions]
    {Conditions for complete problems \\ under polynomial-time kernel reductions}
\label{sec:generalcompleteness}
%
Most well-behaved complexity classes contain problems that are complete under many-one reductions.
Do the corresponding classes of equivalence problems contain problems that are complete under kernel reductions?
Having access to a complete problem offers many benefits and improves our understanding of equivalence problems in general.
In \autocite[Theorem~8.7]{bcffm}, the authors constructed a complete problem with respect to polynomial-time kernel reductions for $\NPEq$ under the assumption that $\NP = \coNP$.
Since we consider that assumption unlikely, we determine sufficient conditions for having a complete problem under polynomial-time kernel reductions.
This section presents a more general theorem that implies as a corollary a complete problem for $\NPEq$ under the assumption $\NP = \coNP$.

%
By extending the technique of \autocite[Theorem~8.7]{bcffm}, we find that $\PSPACEEq$ has a complete problem under polynomial-time kernel reductions unconditionally.
We also show that each level of the polynomial-time hierarchy contains an equivalence problem that is hard for the lower levels under these reductions.
This means that some well-known classes do have complete problems, and the existence for complete problems in other classes, like $\NP$ and even $\P$, remains possible.
The existence of a natural complete problem remains open.

We need one additional definition in order to describe the complexity classes that contain a hard problem under kernel reductions.
If $\mathcal{C}$ is a complexity class then the class $\forall\mathcal{C}$ is the set of languages $A$ such that there exists a language $B\in\mathcal{C}$ and a polynomial $p$ satisfying $x\in A$ if and only if $\forall w\in\Sigma^{\leq p(|x|)} \pair{x}{w}\in B$.
$\forall\mathcal{C}$ is called the \defn{closure of $\mathcal{C}$ under polynomially bounded universal quantification}.

\begin{theorem}\label{thm:generalcompleteness}
  Let $\mathcal{C}$ be a subset of $\PSPACE$ which contains the problem of deciding whether two strings are equal.
  Then there exists an equivalence relation in $\CRAZYEq$ which is hard for $\CEq$ under $\kr$ reductions.
\end{theorem}

Before proving this theorem, we will provide some immediate corollaries of this general result.

\begin{corollary}\label{cor:sufficient}
  If $\mathcal{C}$ is a subset of $\PSPACE$ and $\mathcal{C}=\CRAZY$, then $\CEq$ has a complete problem under $\kr$ reductions.
\end{corollary}

\begin{corollary}\label{cor:hardproblems}
  Under polynomial-time kernel reductions,
  \begin{enumerate}
  \item $\PSPACEEq$ has a complete problem,
  \item $\PKPEq$ contains a problem that is hard for $\DKPEq$, for all $k \geq 1$.
  \end{enumerate}
\end{corollary}
\begin{proof}\mbox{}
  \begin{enumerate}
  \item $\PSPACE$ is closed under complement (because it is a deterministic complexity class) and polynomially bounded universal quantification (because we can simulate the universal guess deterministically in polynomial space).
  \item
    First, $(\forall(\DKP \cup \mathsf{co}\DKP))\mathsf{Eq} = (\forall \DKP)\mathsf{Eq}$, since $\DKP$ is closed under complement.  
    Next, $(\forall \DKP)\mathsf{Eq} = \PKPEq$, since $\forall \DKP = \PKP$.  
    Now if we choose $\mathcal{C} = \DKP$ in \autoref{thm:generalcompleteness}, then $\PKPEq$ has a problem that is hard for $\DKPEq$ under $\kr$ reductions.
    \qedhere
  \end{enumerate}
\end{proof}

More specifically, this means that $\coNPEq$ (which equals $\POPEq$) has a problem that is $\kr$-hard for $\PEq$ (which equals $\DOPEq$).
This corollary also leads to \autocite[Theorem~8.7, part~1]{bcffm}, which is restated here.

\begin{corollary}[{\autocite[Theorem~8.7, part~1]{bcffm}}]
  If $\NP = \coNP$ then $\NPEq$ has a complete problem under polynomial-time kernel reductions.
\end{corollary}
\begin{proof}
  If $\NP = \coNP$, then the polynomial hierarchy collapses to $\POP$, and specifically $\PTP = \DTP = \POP = \coNP = \NP$.
  From \autoref{cor:hardproblems} we conclude that $\NPEq$ has a $\kr$-hard problem for $\NPEq$.
  Such a problem is by definition $\NPEq$-complete.
\end{proof}

We now return to the proof of \autoref{thm:generalcompleteness} by first providing some motivating ideas.
Recall the canonical complete problem (sometimes called the ``universal'' problem) for $\NP$ (and indeed for various other complexity classes):
\begin{displaymath}
  K = \lb\triple{M}{x}{1^t} \st M\plain{accepts} x \plain{within} t \textnormal{ steps}\rb
\end{displaymath}
The idea of this proof is to adapt this into an equivalence relation $R_K$ consisting of pairs of triples of the form $\pair{\triple{M}{x}{1^{t_x}}}{\triple{M}{y}{1^{t_y}}}$, where $M$ accepts $\pair{x}{y}$, as in the reduction from an arbitrary $\NP$ language to $K$.
The problem we encounter here is that $R_K$ is not necessarily an equivalence relation.
Consider, for example, transitivity, which must be satisfied for all possible pairs of the form $\triple{M}{w}{1^{t_w}}$.
For \emph{arbitrary machines} $M$, just because $M$ accepts $\pair{x}{y}$ and $\pair{y}{z}$ does not necessarily mean that $M$ accepts $\pair{x}{z}$.
The solution is to encode into $R_K$ the requirement that the language which $M$ accepts, $L(M)$, is itself an equivalence relation.
The three properties required of $R_K$ then follow from the properties of $L(M)$.
%
%
\begin{proof}[Proof of \autoref{thm:generalcompleteness}]
  First we will define a helper algorithm which decides whether a given machine accepts an equivalence relation on strings up to a given length.
  Define the algorithm $A$ as follows on input $\pair{M}{n}$, where $M$ is a polynomially clocked Turing machine of type $\mathcal{C}$ and $n\in\mathbb{N}$:
  \begin{enumerate}
  \item universally guess $a,b,$ and $c\in\Sigma^{\leq n}$,
  \item simulate $M$ on $\pair{a}{a}$; if it rejects, reject,
  \item simulate $M$ on $\pair{a}{b}$, then on $\pair{b}{a}$; if the former accepts and the latter rejects, reject,
  \item simulate $M$ on $\pair{a}{b}$, then on $\pair{b}{c}$, then on $\pair{a}{c}$; if the first two accept and the last one rejects, reject,
  \item if execution reaches this point, accept.
  \end{enumerate}
  These simulations check that $L(M)$ satisfies reflexivity, symmetry, and transitivity on strings of length at most $n$.
  If $A$ accepts, then the three properties are satisfied, and if it rejects then one of the three properties is violated.
  Since $M$ is a machine of type $\mathcal{C}$, checking if $M$ accepts on some input and if $M$ rejects on some input is in $\mathcal{C}\cup\mathsf{co}\mathcal{C}$.
  The universal guesses of $a,b,$ and $c$ (of length at most $n$) followed by checks of whether the six simulations of $M$ accept or reject place $L(A)$ in the class $\CRAZY$.
  If $p$ is the polynomial which bounds the running time of $M$, then the running time of this algorithm is $6p\left(\left|\pair{1^n}{1^n}\right|\right)+c$, where $c$ is a constant which represents the time needed to account for the implementation of $A$ (the control of the simulations of $M$, performing logical conjunctions, etc.).
  Hence the running time of $A$ is polynomial in $n$.

  Now we can define the set $R_K$ as follows.
  A pair of strings $\pair{u}{v}$ is in $R_K$ if and only if either $u = v$ or $u$ and $v$, when interpreted as strings of the form $\triple{M}{x}{1^{t_x}}$ and $\triple{M}{y}{1^{t_y}}$, respectively, satisfy the four conditions
  \begin{enumerate}
  \item\label{itm:machine} $M$ is a polynomially clocked Turing machine of type $\mathcal{C}$,
  \item\label{itm:emx} $A$ accepts $\pair{M}{|x|}$ within $t_x$ steps,
  \item\label{itm:emy} $A$ accepts $\pair{M}{|y|}$ within $t_y$ steps,
  \item\label{itm:accepts} $M$ accepts $\pair{x}{y}$.
  \end{enumerate}
  We claim that $R_K$ is in $\CRAZYEq$ and $\CEq$-hard.

  First we show that $R_K\in\CRAZY$.
  By the argument above, $A$ is a $\CRAZY$ algorithm.
  Assuming without loss of generality that $|x|\geq |y|$, if $A$ accepts $\pair{M}{|x|}$ within $t_x$ steps then we know that there is a polynomial-time bound on the running time of $M$ on input $\pair{x}{y}$, so simulating it is certainly in $\CRAZY$.
  Finally, testing for equality is in $\mathcal{C}$ by hypothesis so deciding $R_K$ overall can be performed by a $\CRAZY$ algorithm.

  Next we show that $R_K$ is an equivalence relation.
  Reflexivity follows from the reflexivity of the equality relation.
  For symmetry, suppose that the pair $\pair{\triple{M}{x}{1^{t_x}}}{\triple{M}{y}{1^{t_y}}}$ is in $R_K$.
  Since \autoref{itm:emx} and \autoref{itm:emy} are true by hypothesis, we know that symmetry on strings of length at most $\max(|x|, |y|)$ in $L(M)$ is satisfied, and that includes the strings $x$ and $y$.
  So since $M$ accepts $\pair{x}{y}$ it must follow that $M$ accepts $\pair{y}{x}$.
  Furthermore, \autoref{itm:machine}, \autoref{itm:emx}, and \autoref{itm:emy} are the same up to symmetry of $x$ and $y$, so we have $\pair{\triple{M}{y}{1^{t_y}}}{\triple{M}{x}{1^{t_x}}}\in R_K$.
  For transitivity, suppose that both $\pair{\triple{M}{x}{1^{t_x}}}{\triple{M}{y}{1^{t_y}}}\in R_K$ and $\pair{\triple{M}{y}{1^{t_y}}}{\triple{M}{z}{1^{t_z}}}\in R_K$.
  Since transitivity is true on strings of length at most $\max(|x|, |y|, |z|)$ by the transitivity propositions checked by \autoref{itm:emx} and \autoref{itm:emy}, and since $M$ accepts both $\pair{x}{y}$ and $\pair{y}{z}$ by hypothesis, it must follow that $M$ accepts $\pair{x}{z}$.
  Again the conditions in \autoref{itm:machine}, \autoref{itm:emx}, and \autoref{itm:emy} are the same.
  We have shown that $R_K$ is reflexive, symmetric, and transitive, so it is an equivalence relation.
  At this point, we have proven that $R_K\in\CRAZYEq$.

  Now we need to show that $R_K$ is $\CEq$-hard.
  Let $S\in\CEq$.
  Suppose $M$ is the polynomially clocked $\mathcal{C}$ machine that decides $S$, and $p$ is the polynomial that bounds the running time of $M$.
  Then the kernel reduction from $S$ to $R_K$ is $w\mapsto\triple{M}{w}{1^{6p(|\pair{w}{w}|)+c}}$, where $p$ and $c$ are the polynomial and constant described in the first paragraph of this proof.
  Call this reduction $f$.
  The reduction is obviously computable in time polynomial in $|w|$.
  It remains to show that this reduction is correct.

  Suppose $\pair{x}{y}\in S$.
  Now $f(x)=\triple{M}{x}{1^{6p(|\pair{x}{x}|)+c}}$ and, similarly, $f(y)=\triple{M}{y}{1^{6p(|\pair{y}{y}|)+c}}$.
  \autoref{itm:machine} is true by construction, and \autoref{itm:accepts} is true since $M$ is the machine which decides $S$.
  Assume \autoref{itm:emx} is false.
  Then $M$ does not accept an equivalence relation on strings of length at most $|x|$.
  This is a contradiction, since $M$ decides $S$, an equivalence relation, by hypothesis.
  Therefore \autoref{itm:emx} must be satisfied.
  The same argument applies to \autoref{itm:emy}.
  Hence $\pair{f(x)}{f(y)}\in R_K$.

  If $\pair{x}{y}\notin S$ then $M$ does not accept $\pair{x}{y}$, since otherwise $\pair{x}{y}$ would be a member of $S$.
  Hence $\pair{x}{y}\notin R_K$.
  Therefore we have shown that $R_K$ is $\CEq$-hard.
\end{proof}

\begin{openproblem}
  Is there a more general characterization of complexity classes which have a $\kr$-hard problem?
\end{openproblem}

\begin{openproblem}\label{open:npeqc}
  Is there an oracle relative to which $\PEq$ or $\NPEq$ has a complete problem under polynomial-time kernel reductions?
  We conjecture that $\NPEq$ has a complete problem without relativization.
\end{openproblem}

\begin{openproblem}
  Is the converse of \autoref{cor:sufficient}, or perhaps a partial converse, true?
  In other words, is it true that the existence of a $\CEq$-complete problem problem implies closure under complement or universal quantification (or both)?
  If so, this would be evidence that no $\NPEq$-complete problem exists, since this would imply $\NP = \coNP$.
\end{openproblem}

\begin{openproblem}
  Can this theorem be used to construct $\krnt$-hard problems for smaller complexity classes such as $\NLEq$ under the appropriate time-bounded reduction?
  Larger classes such as $\EXPEq$?
\end{openproblem}

\begin{openproblem}
  To what other equivalence relations does our $\kr$-hard problem reduce?
  Are there ``natural'' $\kr$-hard problems in complexity classes which satisfy the conditions in \autoref{thm:generalcompleteness}?
\end{openproblem}

\section
    [Relationship between completeness under kernel and many-one reductions]
    {Relationship between completeness \\ under kernel and many-one reductions}
\label{sec:npeqcompleteness}
%
A kernel reduction implies a many-one reduction, but does completeness under kernel reductions imply completeness under many-one reductions?
Since polynomial-time kernel reductions are different from polynomial-time many-one reductions (\autoref{thm:different}), completeness in classes of equivalence problems may differ under these reductions as well.
We determine the conditions under which completeness under kernel reductions implies completeness under many-one reductions.

%
We find that completeness under many-one reductions follows as a straightforward consequence of completeness under kernel reductions as long as the relevant complexity class admits a complete problem under many-one reductions.
We also show that the kernel reduction is essentially too weak to allow for completeness under injective (that is, ``one-to-one'') reductions, for combinatorial reasons similar to those in \autoref{sec:limitations}.
Though we prove these results for $\NPEq$, they generalize in a natural way to any ``well-behaved'' complexity class (basically, any class containing a complete problem under many-one reductions).
These results are more indication that when comparing the relative difficulty of equivalence problems, one should attempt to construct a kernel reduction instead of a many-one reduction.
The potential lack of a complete problem under injective kernel reductions suggests that a conjecture analagous to the Berman--Hartmanis conjecture, which states that all $\NP$-complete problems are isomorphic with respect to many-one reductions, may be false in $\NPEq$.

One can infer the existence of $\NP$-complete equivalence relations from the relation suggested in \autocite[Section~6.2]{fg11},
\begin{equation*}
  \{\pair{0\phi}{1\phi} \, | \, \phi \in \textsc{Satisfiability}\}.
\end{equation*}
(This relation is not itself an equivalence relation, but can be modified to guarantee the three necessary properties.)
Using this idea, we provide a strategy for constructing a more natural $\NP$-complete equivalence relation from an equivalence relation in $\NP$ and an arbitrary $\NP$-complete property.

Let \textsc{GI} denote the equivalence relation consisting of all pairs of isomorphic graphs.
A property, that is, a Boolean function, $\Pi$ is an \defn{$\NP$-complete property} if $L_\Pi$, the set of all strings for which $\Pi$ is true, is $\NP$-complete.
If, furthermore, the property satisfies $\pair{x}{y} \in R$ implies $\Pi(x) = \Pi(y)$ where $R$ is an equivalence relation, $\Pi$ is called a \emph{property on $R$}.
For example, Hamiltonicity, the property of having a cycle that includes each vertex, is an $\NP$-complete property on \textsc{GI}.

\begin{theorem}\label{thm:npceqrel}
  If $\Pi$ is an \NP-complete property on \textsc{GI}, then the equivalence relation $A$ defined by
  \begin{equation*}
    A = \left\{ \pair{G}{H} \, \middle| \, \pair{G}{H} \in \textsc{GI} \text{ or } (G \in L_\Pi \text{ and } H \in L_\Pi) \right\}
  \end{equation*}
  is an \NP-complete equivalence relation.
\end{theorem}
\begin{proof}
  It is straightforward to prove that $A$ is an equivalence relation, so it remains to show that it is $\NP$-complete.
  The language $A$ is in $\NP$ because both $R$ and $L_\Pi$ are in $\NP$ by hypothesis.
  Thus we need only show that $A$ is $\NP$-hard.

  Let $H$ be a graph satisfying $\Pi$; such a graph must exist because $\Pi$ is $\NP$-complete and therefore there must be at least one graph that satisfies $\Pi$ and at least one that does not (otherwise no many-one reduction to $L_\Pi$ could exist).
  The reduction proving that $A$ is $\NP$-complete is from $L_\Pi$, and the mapping is given by $G \mapsto \pair{G}{H}$.
  This function is computable in linear time; the size of $H$ is constant with respect to the size of $G$.

  Now we show that $G \in L_\Pi$ if and only if $\pair{G}{H} \in A$, for any graph $G$.
  If $G \in L_\Pi$, then $G \in L_\Pi$ and $H \in L_\Pi$, so $\pair{G}{H} \in A$.
  If $\pair{G}{H} \in A$, then either $G \in L_\Pi$ and $H \in L_\Pi$, in which case $G \in L_\Pi$, or $G$ is isomorphic to $H$, in which case $G$ is in $L_\Pi$ because $H$ is.
  In either case $G \in L_\Pi$.
  We conclude that $L_\Pi \mor A$, and so $A$ is an $\NP$-complete equivalence relation.
\end{proof}

\begin{example}\label{ex:npceqrel}
  The language
  \begin{equation*}
    \left\{ \pair{G}{H} \, \middle| \, \pair{G}{H} \in \textsc{GI} \text{ or } G \text{ and } H \text{ have a Hamiltonian cycle}  \right\}
  \end{equation*}
  is an \NP-complete equivalence relation.
\end{example}

There are other ways of constructing a natural $\NP$-complete equivalence relation.
For example, there is a finitely presented group whose word problem is $\NP$-complete \autocite[Corollary~1.1]{sbr02}, and the word problem is already an equivalence relation.
This may be considered ``more natural'' because it does not involve the disjunction of two distinct computational problems, though it lacks the simplicity of our approach.

\begin{example}
  As stated briefly above, \autoref{thm:npceqrel} can be generalized to isomorphism of structures other than graphs and/or larger complexity classes.
  For example, replacing \textsc{GI} with \textsc{FI}, the Boolean formula isomorphism problem, and an $\NP$-complete property on \textsc{GI} with a $\SigmaTwoP$-complete property on \textsc{FI} yields a $\SigmaTwoP$-complete equivalence relation.
\end{example}

\begin{corollary}
  If $R$ is $\NPEq$-complete, then $R$ is $\NP$-complete.
\end{corollary}
\begin{proof}
  This follows immediately from the existence of an $\NP$-complete equivalence relation, as in \autoref{ex:npceqrel}, and the fact that a kernel reduction implies a many-one reduction.
\end{proof}

This corollary provides a clearer proof of \autocite[Proposition~8.1]{bcffm}.

\begin{corollary}[{\autocite[Proposition~8.1]{bcffm}}]
  If \textsc{GI} is \NPEq-complete then the polynomial hierarchy collapses to the second level, that is, $\PH = \STP$.
\end{corollary}
\begin{proof}
  By the previous corollary, if \textsc{GI} is $\NPEq$-complete, then it is $\NP$-complete, which implies the stated collapse (see \autocite{schoning87}).
\end{proof}

\autoref{thm:npceqrel} also provides a simple method for proving the equivalence of $\P = \NP$ and $\PEq = \NPEq$.

\begin{theorem}\label{thm:pnppeqnpeq}
  $\P = \NP$ if and only if $\PEq = \NPEq$.
\end{theorem}
\begin{proof}
  If $\P = \NP$, then $\PEq = \NPEq$ by their definitions.
  Suppose now that $\PEq = \NPEq$.
  Let $A$ denote the $\NP$-complete equivalence relation defined in \autoref{ex:npceqrel}.
  Since $A \in \NPEq$ and $\PEq = \NPEq$ by hypothesis, $A \in \PEq$, and hence $A \in \P$.
  Since $\P$ is closed under $\mor$ reductions, any $\NP$-complete problem in $\P$ implies $\P = \NP$.
\end{proof}

As stated in \autoref{open:npeqc}, we do not know whether an $\NPEq$-complete problem exists.
In the following theorem we describe an equivalence relation that, if it were $\NPEq$-complete, would prove that injective kernel reductions are strictly weaker than general kernel reductions.
This is interesting because it again demonstrates that the number and size of equivalence classes is important when considering the (im)possibility of polynomial-time kernel reductions between equivalence relations.
In the following theorem, if an equivalence relation is ``complete under $\kri$ reductions in $\NPEq$'' we mean that every equivalence relation in $\NPEq$ reduces to it by a polynomial-time computable kernel reduction which is also injective (that is, ``one-to-one'').

\begin{theorem}\label{thm:inj}
  Let $\Pi$ be a property on \textsc{GI}.
  If the equivalence relation $A$ defined by
  \begin{equation*}
    A = \left\{ \pair{G}{H} \, \middle| \, \pair{G}{H} \in \textsc{GI} \text{ or } (G \in L_\Pi \text{ and } H \in L_\Pi \text{ and } |G| = |H|) \right\}
  \end{equation*}
  is complete for $\NPEq$ under $\kr$ reductions, then $A$ is not complete under $\kri$ reductions.
\end{theorem}

The only difference between the equivalence relation $A$ defined here and the one defined in \autoref{thm:npceqrel} is the requirement that $|G| = |H|$.
This means that although the number of equivalence classes in $A$ is infinite (at least one for each size), each of those equivalence classes is itself finite.
In contrast, consider the equivalence relation $S$ defined by
\begin{equation*}\label{eq:ones}
  S = \left\{ \pair{x}{y} \, \middle| \, x \text{ and } y \text{ have the same number of } 1 \text{s} \right\}.
\end{equation*}
The equivalence relation $S$ has an infinite number of equivalence classes: $[1]$, $[11]$, $[111]$, etc.
Each equivalence class is itself infinite as well: for each $w \in \Sigma^*$, the equivalence class $[w]$ contains $w$, $0w$, $00w$, etc.

\begin{proof}[Proof of \autoref{thm:inj}]
  Let $S$ be the equivalence relation defined in the preceding paragraph.
  The language $S$ is decidable in linear time by a deterministic Turing machine, hence it is in $\NP$.
  Since $A$ is $\NPEq$-complete by hypothesis, $S \kr A$.
  Thus there is a polynomial-time computable function $f$ such that $\pair{x}{y} \in S$ if and only if $\pair{f(x)}{f(y)} \in A$.

  By the discussion preceding this theorem, $[w]_S$ is infinite and $[f(w)]_A$ is finite.
  By \autoref{lem:image}, $f([w]_S) \subseteq [f(w)]_A$.
  Consider $f|_{[w]_S}$, that is, $f$ restricted to the domain $[w]_S$.
  Then $f|_{[w]_S}$ is a mapping from the infinite set $[w]_S$ to the finite set $[f(w)]_A$.
  By the pigeonhole principle, $f|_{[w]_S}$ is not injective.
  Hence the unrestricted reduction $f$ is not injective, and therefore $A$ is not $\kri$-complete in $\NPEq$.
\end{proof}

\newcommand{\floor}[1]{\lfloor #1 \rfloor}
\section{Existence of intermediary problems}
\label{sec:intermediary}
%
According to the seminal theorem by Ladner \autocite{ladner75}, if $\P \neq \NP$, then there are problems of intermediate complexity, in the sense that these problems are neither in $\P$ nor $\NP$-complete.
The theorem does not immediately imply a similar result for equivalence problems, since $\PEq$ is different from $\P$ and $\NPEq$ is different from $\NP$ (specifically, in each case, the latter contains problems that are not equivalence problems).
Do kernel reductions induce the same rich structure between $\PEq$ and $\NPEq$ as do many-one reductions between $\P$ and $\NP$?
We adapt a proof of Ladner's theorem from \autocite{df03} (which has been attributed to Russell~Impagliazzo) to classes of equivalence problems;
this section details that adaptation.

%
The main theorem of this section is the existence of $\NPEq$-intermediary problems under the assumption that $\PEq \neq \NPEq$ (which is equivalent to the assumption $\P \neq \NP$ by \autoref{thm:pnppeqnpeq}).
We conclude that even though kernel reductions are strictly weaker than many-one reductions, they still preserve the hierarchies of problems of various computational complexities we expect from our understanding of traditional complexity classes.
The graph isomorphism problem, as one of the few candidates for an $\NP$-intermediary problem, may be the best candidate for a natural $\NPEq$-intermediary problem as well.

This proof of Ladner's theorem for equivalence relations is a delayed diagonalization via progressive padding.
First we define the equivalence relation performing the diagonalization and the corresponding padding function, then we show that this problem is neither in $\PEq$ nor $\NPEq$-complete.

\begin{construction}\label{con:diag}
  Let $K$ be an $\NP$-complete equivalence relation.
  We know that such equivalence relations exist by \autoref{thm:npceqrel}.
  Define the equivalence relation $R$ by
  \begin{equation*}
    R = \left\{\left\langle x 0 1^{p(n) - n - 1}, y 0 1^{p(n) - n - 1}\right\rangle \, \middle| \, \pair{x}{y} \in K \text{ and } |x| = |y| = n\right\},
  \end{equation*}
  where $p$ is a padding function that will be defined below.
  The equivalence relation $R$ is a padded version of $K$.

  Our goal is to define the function $p$ so that $R$ is not too hard and not too easy: it's output should be large enough that $R$ is not $\NPEq$-complete but not so large that $R$ is in $\P$.
  For this we need an enumeration of each polynomially clocked Turing machine, $\{M_i\}_i$, where machine $M_i$ halts within time $i n^i$ on inputs of length $n$.
  For any pair of strings $x$ and $y$, we say a Turing machine $M$ \emph{disagrees with $R$ on $\pair{x}{y}$} if
  \begin{itemize}
  \item $M(\pair{x}{y})$ accepts and $\pair{x}{y} \notin R$, or
  \item $M(\pair{x}{y})$ rejects and $\pair{x}{y} \in R$.
  \end{itemize}
  We define $p$ for each positive integer $n$ by the following iterative process (and thus we implicitly define $R$ iteratively as well).
  Initially, let $i = 1$, then perform the following steps for each $n$ in order.
  \begin{itemize}
  \item Define $p(n)$ to be $n^i$.
  \item
    Check if there is any pair of strings $x$ and $y$, each of length at most $\log \log n$, such that $M_i$ disagrees with $R$ on $\pair{x}{y}$.
    If any such pair exists \emph{and} $\floor{\log \log n}$ is an integer not already seen, then increment $i$.\qedhere
  \end{itemize}
\end{construction}

The following three lemmas prove that this problem is of intermediate complexity if $\PEq \neq \NPEq$.

\begin{lemma}
  The function $p$ in \autoref{con:diag} is computable in time polynomial in $n$.
\end{lemma}
\begin{proof}
  Computing $p(n)$ requires computing $p(1)$, $p(2)$, $\dotsc$, $p(n - 1)$; if each of these $n - 1$ computations takes a polynomial amount of time, the total time required to compute $f(n)$ remains polynomial in $n$, by induction.
  Since the strings $x$ and $y$ are of length at most $\log \log n$, the total number of iterations required to test all pairs of strings is polynomial in $n$.
  The simulation of $M_i$ is computable in time $i (\log \log n)^i$, but $i$ is at most $\log \log n$, since $i$ can only be incremented at most $\log \log n$ times.
  Using the fact that a polynomial in $\log \log n$ is bounded above by $O(\log n)$,
  \begin{align*}
    i (\log \log n)^i & \leq \log \log n (\log \log n)^{\log \log n} \\
    & = 2^{(\log \log \log n)^2 \log \log n } \\
    & \leq 2^{(\log \log n)^2} \\
    & = 2^{O(\log n)} \\
    & = \poly(n),
  \end{align*}
  so the machine $M_i$ runs in time polynomial in $n$.
  The language $R$ is in $\NPEq$ because it is a padded version of the language $K$, which is in $\NPEq$.
  Since $\NP \subseteq \EXP$ and the inputs $x$ and $y$ are each of length $\log \log n$, membership in $R$ can be determined in time polynomial in $n$.
  (Even though the definition of $R$ requires $p$ to be defined, $p$ is already defined for strings of length less than $n$, including the strings $x$ and $y$.)
  Since each step can be performed in polynomial time and there are at most $n$ iterations required when defining $p(n)$, we conclude that $p$ is computable in time polynomial in $n$.
\end{proof}

\begin{lemma}
  Suppose $R$ is the equivalence relation in \autoref{con:diag}.
  If $\PEq \neq \NPEq$, then $R \notin \PEq$.
\end{lemma}
\begin{proof}
  Assume with the intention of producing a contradiction that $R \in \PEq$.
  Thus there is a natural number $i$ such that $M_i$ decides $R$.
  For sufficiently large $n$, the machine $M_i$ never disagrees with $R$, so $p(n) = n^i$ for all sufficiently large $n$.
  Assuming without loss of generality that the string $x$ and $y$ are each of length $n$, this yields a polynomial-time kernel reduction from $K$ to $R$ via the function $\pair{x}{y} \mapsto \left\langle x 0 1^{p(n) - n - 1}, y 0 1^{p(n) - n - 1}\right\rangle$.
  This mapping is polynomial-time computable because $p(n) = n^i$ for all sufficiently large $n$, and $i$ does not depend on $n$.
  Since $\PEq$ is closed under polynomial-time kernel reductions, $K$ is in $\PEq$, and hence $\PEq = \NPEq$, since $K$ is $\NPEq$-complete.
  This is a contradiction with the assumption that $\PEq \neq \NPEq$, hence $R \notin \PEq$.
\end{proof}

\begin{lemma}
  Suppose $R$ is the equivalence relation in \autoref{con:diag}.
  If $\PEq \neq \NPEq$, then $R$ is not $\NPEq$-complete.
\end{lemma}
\begin{proof}
  Assume with the intention of producing a contradiction that $R$ is $\NPEq$-complete.
  Thus $K \kr R$, so there is a function $f$ such that $f$ halts within $n^j$ steps and for each string $x$ and $y$, we have $\pair{x}{y} \in K$ if and only if $\pair{f(x)}{f(y)} \in R$.
  If the image of $f$ were finite, then $R$ would have a constant number of equivalence classes.
  In this case, $R$ would be in $\PEq$, and since $\PEq$ is closed under polynomial-time kernel reductions, $K$ would be in $\PEq$ as well, a contradiction with the hypothesis that $\PEq \neq \NPEq$.

  Suppose the image of $f$ is infinite.
  We know $f(w)$ must be of the form $w01^{p(|w|) - n - 1}$ for each string $w$ of length $n$, so the length of $f(w)$ is $p(|w|)$.
  There is a natural number $n_0$ such that for each $n \geq n_0$, there is a positive integer $k$ such that $k$ is greater than $j$ and for each string $w$ of length $n$, we have $|f(w)| = p(n) = n^k$.
  (The integer $k$ is strictly greater than $j$, since if it were less than or equal to $j$, the image of $f$ would be finite.)
  Now we can construct a polynomial-time algorithm for $K$.
  Assume without loss of generality that all inputs are pairs of strings of equal length.
  On inputs of the form $\pair{x}{y}$, proceed as follows.
  \begin{itemize}
  \item If $|x| < n_0$ (or equivalently $|y| < n_0$), decide whether $\pair{x}{y} \in K$ by examining a hardcoded lookup table for strings of length less than $n_0$.
  \item Compute $f(x)$ and $f(y)$.
  \item If either $|f(x)|$ or $|f(y)|$ is not in the range of $p$, reject.
  \item
    Suppose $f(x) = x' 0 1^{p(m) - m - 1}$ and $f(y) = y' 0 1^{p(m) - m - 1}$, where $|x'| = |y'| = m$.
    Invoke this algorithm recursively on input $\pair{x'}{y'}$.
  \end{itemize}

  Assuming for now that $x'$ and $y'$ are shorter than $x$ and $y$.
  Then the correctness of this algorithm follows from the fact that
  \begin{equation*}
    \pair{x}{y} \in K \iff \pair{f(x)}{f(y)} \in R \iff \pair{x'}{y'} \in K.
  \end{equation*}
  Since the length of the inputs to the algorithm decrease on each recursive invocation, there are at most $n$ recursive calls on inputs of pairs of strings of length $n$.
  Eventually the solution can be found in the hardcoded lookup table (the base case of the recursion).
  Each recursive invocation of the algorithm other than the base case requires computing $f$ on an input of length $n$ (twice), which can be done in polynomial time.
  Thus the overall time required for this algorithm is polynomial in $n$.
  This proves that $K \in \P$ and thus $\P = \NP$.
  Since $\P = \NP$ if and only if $\PEq = \NPEq$, we have a contradiction.

  Finally, we prove that $|x'| < |x|$ (the proof that $|y'| < |y|$ is the same), which we postponed from the previous paragraph.
  Due to its time bound, $|f(x)| \leq |x|^j$ for any string $x$.
  By assumption, $p(|x'|) = |x'|^k$.
  By construction, $p(|x'|) = |f(x)|$
  Combining these three relations yields the inequality
  \begin{equation*}
    |x'|^k = p(|x'|) = |f(x)| \leq |x|^j,
  \end{equation*}
  so $|x'| \leq |x|^{j / k} < |x|$, since $k > j$ and lengths must be natural numbers.
\end{proof}

Combining the preceding three lemmas yields Ladner's theorem for classes of equivalence relations.

\begin{theorem}\label{thm:intermediary}
  If $\PEq \neq \NPEq$, then there is an equivalence relation in $\NPEq$ that is neither in $\PEq$ nor $\NPEq$-complete.
\end{theorem}

This technique can be generalized to other classes of equivalence relations, as long as the underlying machines for the smaller class can be enumerated and the larger class has an equivalence relation that is complete under many-one reductions.
For example, we can produce equivalence relations between the polynomial hierarchy and $\PSPACE$.

\begin{corollary}\label{cor:pspace}
  If $\PHEq \neq \PSPACEEq$, then there is an equivalence relation in $\PSPACEEq$ that is neither in $\PHEq$ nor $\PSPACEEq$-complete.
\end{corollary}

\section{Conclusion}
%
Throughout this work we have proven that kernel reductions are similar to many-one reductions in the most basic ways, but differ in some key aspects.
Like many-one reductions, kernel reductions are transitive and have good closure properties.
The class of equivalence problems in $\PSPACE$ has a complete problem under kernel reductions (\autoref{cor:hardproblems}).
The equivalence of the two equalities $\P = \NP$ and $\PEq = \NPEq$ (\autoref{thm:pnppeqnpeq}) uses the similarity between many-one and kernel reductions.
Just as many-one reductions allow the existence of $\NP$-intermediary problems, kernel reductions allow for the possibility of $\NPEq$-intermediary problems (\autoref{thm:intermediary}).
On the other hand, there are equivalence relations between which there is a many-one reduction but no kernel reduction (\autoref{thm:different}).
Specifically, if there are more equivalence classes, up to strings of certain lengths, in $R$ than in $S$, then no kernel reduction can exist.
Finally, under some assumptions, there is an equivalence problem that is not complete for $\NPEq$ under injective kernel reductions (\autoref{thm:inj}), whereas nearly every known $\NP$-complete problem is isomorphic (the Berman--Hartmanis conjecture \autocite{bh77} states that \emph{every} $\NP$-complete problem is isomorphic).

The techniques used in this paper to show that kernel reductions are weaker than many-one reductions are combinatorial techniques (for example, comparing the numbers of equivalence classes).
Combining these with other complexity theoretic and algebraic techniques has already proven useful: there is no polynomial-time kernel reduction from the graph isomorphism problem to the isomorphism problem for strongly regular graphs \autocite[Theorem~22]{babai14}.
This is interesting because even though the latter appears to be a difficult problem, no polynomial-time many-one reduction from the former to the latter is expected to exist \autocite{babai14}, and the lack of a kernel reduction is evidence in that direction.

Besides the open problems listed in \autoref{sec:generalcompleteness}, we consider the following questions to be worth exploring.
\begin{itemize}
\item There are several problems of \emph{inequivalence} in \autocite{gj79} listed as $\NP$-complete or as $\PSPACE$-complete.
  What do these problems have to do with $\NPEq$-completeness, $\coNPEq$-completeness, and $\PSPACEEq$-completeness?
\item When can results like \autocite[Theorem~1]{rs11}, for example, which shows that an equivalence relation is complete for $\L$ under many-one reductions via a reduction from a problem that is \emph{not} an equivalence relation, be translated to a proof that the problem is complete under kernel reductions for the corresponding class of equivalence problems?
\item
  In the case of the graph isomorphism problem, \textsc{GI}, the number $\#\textsc{GI}(n)$, in the notation of \autoref{sec:limitations}, is the number of (pairwise) non-isomorphic graphs on at most $n$ vertices.
  This differs from the conventional notation $\#\textsc{GI}$ denoting the problem of counting the number of graphs isomorphic to a given graph.
  In other words, our notation counts the \emph{number} of equivalence classes, whereas the latter counts the \emph{size} of an equivalence class.
  For the graph isomorphism problem, computing the size of an equivalence class is Turing-equivalent to deciding whether two graphs are isomorphic \autocite[Theorem~1.24]{kst93}.
  When is the problem of computing the size of an equivalence class Turing-equivalent to the problem of deciding equivalence?
\end{itemize}

\section{Acknowledgments}

The authors acknowledge the invaluable help provided by Josh~Grochow and Steve~Homer.
We thank the anonymous reviewers of an earlier version of this paper for numerous corrections and stylistic suggestions.
Specifically, we thank an anonymous reviewer for showing how to prove $\Cl_1 \subseteq \NKer_1$, and another anonymous reviewer for simplifying the proof of \autoref{thm:density}.

\printbibliography

\end{document}